\documentclass[11pt]{article}

\usepackage[final]{nips_2016_blah}

\usepackage[utf8]{inputenc} 
\usepackage[T1]{fontenc}    
\usepackage{hyperref}       
\usepackage{url}            
\usepackage{booktabs}       
\usepackage{amsfonts}       
\usepackage{nicefrac}       
\usepackage{microtype}      

\usepackage{graphicx}
\usepackage{color}
\usepackage{subcaption}

\usepackage{paralist}

\usepackage{amssymb,amsmath,amsthm}

\usepackage{theoremref}
\newtheorem{definition}{Definition}
\newtheorem{theorem}{Theorem}
\newtheorem{proposition}{Proposition}
\newtheorem{lemma}{Lemma}
\newtheorem{corollary}{Corollary}
\newtheorem{remark}{Remark}

\newtheorem{problem}{Problem}

\title{Supermodular Locality Sensitive Hashes}

\author{
Maxim Berman \qquad Matthew B.\ Blaschko\\
Center for Processing Speech \& Images\\
Departement Elektrotechniek, KU Leuven\\
Kasteelpark Arenberg 10\\
3001 Leuven, Belgium\\
\texttt{\{maxim.berman,matthew.blaschko\}@esat.kuleuven.be} \\
}

\begin{document}

\maketitle

\begin{abstract}
In this work, we show deep connections between Locality Sensitive Hashability and submodular analysis. 
We show that the LSHablility of the most commonly analyzed set similarities is in one-to-one correspondance with the supermodularity of these similarities when taken with respect to the symmetric difference of their arguments. 
We find that the supermodularity of equivalent LSHable similarities can be dependent on the set encoding. 
While monotonicity and supermodularity does not imply the metric condition necessary for supermodularity, this condition is guaranteed for the more restricted class of supermodular Hamming similarities that we introduce. 
We show moreover that LSH preserving transformations are also supermodular-preserving, yielding a way to generate families of similarities both LSHable and supermodular. 
Finally, we show that even the more restricted family of cardinality-based supermodular Hamming similarities presents promising aspects for the study of the link between LSHability and supermodularity. 
We hope that the several bridges that we introduce between LSHability and supermodularity paves the way to a better understanding both of supermodular analysis and LSHability, notably in the context of large-scale supermodular optimization. 
\end{abstract}

\section{Introduction}

Locality sensitive hashing is a frequently employed scheme for large scale similarity comparison.  Given a similarity that compares two objects, locality sensitive hashing replaces the exact pairwise computation by a distribution over hash functions for which the expected collision rate between the hashes of two objects is equal to their similarity \cite{Broder1997,BRODER2000630}.  The set of functions and associated distribution is called an LSH, and is advantageous as approximation is straightforward by subsampling the hash functions, and fast schemes exist for comparing binary strings.  It is well known that not all conceivable similarity functions admit an LSH, and the characterization of the set of LSHable functions remains incomplete.

In this work, we show deep connections between LSHability and submodular analysis \cite{fujishige2005submodular}, in particular that for the most commonly analyzed similarities, there is a one-to-one relationship between LSHability and supermodularity of similarities between sets when taken with respect to the symmetric difference of their arguments.  We therefore explore more deeply the relationship between LSHability and supermodularity.  

While we have not found an example that disproves the possibility that the concepts of LSHability and supermodularity are essentially equivalent (key result \ref{item:OneToOneLSHabilitySupermodularity}), we have shown that one can provide alternate set encodings of the same similarity where supermodularity of the resulting set similarity is not guaranteed (key result \ref{item:NonUniquenessEncoding}).  Furthermore, supermodularity alone is not sufficient to ensure that the resulting similarity yields a metric (key result \ref{item:supermodularGenerallyNotMetric}), necessitating the addition of metric constraints to the definition of a \emph{metric supermodular similarity}.  On the contrary, we have demonstrated that LSH preserving functions are supermodular preserving functions and one can consider a closed class of similarities obeying both properties (key result \ref{item:LSHpreservingSupermodularPreserving}). We therefore have introduced a more tractable restricted family of supermodular similarities for which metric conditions are always fulfilled (key result \ref{item:SupermodularHammingSimilarities}).  Finally, we then use this family to identify a non-trivial set of supermodular similarities characterized by convex functions for which known LSHable functions form a strict subset (key result \ref{item:CSHS}).

Testing the submodularity of set functions and the LSHability of similarities are both NP hard problems \cite{CHIERICHETTI201489,seshadhri2014submodularity}. However efficient schemes, independent of the domain size, exist for testing $\varepsilon$-approximate submodularity in the $\ell_p$ sense for $p \geq 1$ \cite{blais2016testing}. To the best of our knowledge, similar results do not exist concerning $\varepsilon$-approximate testing of LSH-ability. 
This observation further motivates our present contribution, in the hope that the connections between LSHability and supermodularity lead to novel approximation schemes for LSHability.

More immediate to practical applications are that both submodularity and LSHability are frequently motivated by approximations and optimization schemes for similar applications. These include near duplicate detection by similarity maximization using a LSH \cite{Ke:2004:EPN:1027527.1027729} and the related problem of diverse $k$-best, which corresponds to supermodular similarity minimization \cite{NIPS2014_5415}.  Many clustering algorithms, including $k$-means, use distance metrics in their formulation \cite{Deza2009}, which can then utilize submodular properties in their optimization, or be approximated by an LSH to achieve sub-quadratic solutions.  In these and related settings, recognizing that LSHability and supermodularity are at least frequently co-occurring properties indicates the possibility to exploit both simultaneously, yielding more efficient and accurate algorithms.

\subsection{Summary of Key Results}

In this paper, we pose the question of the precise relationship between LSHability and supermodularity, showing an intriguing overlap between the two properties.  We additionally show a number of results that elucidate some boundaries of the relationship between LSHability and supermodularity, while demonstrating promising avenues to advance the theory of the relationship between LSHability and supermodularity:
\begin{enumerate}
\item \label{item:OneToOneLSHabilitySupermodularity} A one-to-one relationship between LSHability (\thref{def:LSHability}) and supermodularity (\thref{def:SupermodularSimilarity}) for the most commonly studied similarities (Table~\ref{tab:setSimilarityMeasures});
\item \label{item:NonUniquenessEncoding} The non-uniqueness of a set encoding can lead to equivalent LSHable similarities being supermodular or not (\thref{thm:IntersectionSimilarity} and \thref{thm:IdentityIntersectionSimilarity}; \thref{thm:EncodingSupermodularOrNot});
\item \label{item:SupermodularHammingSimilarities} A construction of supermodular Hamming similarities, a class of supermodular similarities for which metric conditions are always guaranteed (\thref{thm:SupermodularHammingSimilarityMetric}), from arbitrary submodular or supermodular functions  (\thref{thm:SupermodularHammingSimilarityConstruction});
\item \label{item:supermodularGenerallyNotMetric} A proof that supermodularity (\thref{def:SupermodularSimilarity}) is not sufficient to guarantee that a similarity yields a metric (\thref{thm:SupermodularityNotImpliesMetric}), which is a necessary condition for LSHability, and a subsequent definition of metric supermodular similarities (\thref{def:MetricSupermodularSimilarity});
\item \label{item:LSHpreservingSupermodularPreserving} LSH-preserving functions \cite{Chierichetti:2012:LFA:2095116.2095201} are supermodular-preserving functions  (Section~\ref{sec:LSHpreservingIsSupermodularPreserving});
\item \label{item:CSHS} Cardinality-based supermodular Hamming similarities (characterized by a single convex function) are a strict superset of LSH-preserving functions applied to the Hamming similarity, providing a promising entry-point to answering deep questions about the relationship between LSHability and submodularity on the basis of convex analysis (Section~\ref{sec:SymmetricSupermodularHammingSimilarities}).
\end{enumerate}

The rest of the paper is organized as follows: We formally introduce open problems in Section~\ref{sec:OpenProblems}.
We then provide an introduction to submodularity (Section~\ref{sec:Submodularity}), and demonstrate its relationship to commonly employed set similarities (Section~\ref{sec:SubmodularityofSetSimilarity}).  Interestingly, a one-to-one relationship emerges between supermodular similarities and LSHable similarities.
Metric properties of supermodular similarities are discussed in Section~\ref{sec:MetricSupermodularSimilarities}, showing that supermodular Hamming similarities yield a metric, while metric conditions do not necessarily hold for more general submodular similarities.
We subsequently demonstrate in Section~\ref{sec:LSHpreservingIsSupermodularPreserving} that LSH-preserving functions are supermodularity-preserving functions.  Finally, we introduce cardinality-based supermodular Hamming similarities (Section~\ref{sec:SymmetricSupermodularHammingSimilarities}) as a method for employing convex analysis to illuminate the relationship between supermodulariity and LSHability.

\subsection{Open Problems}\label{sec:OpenProblems}

Consider the set of LSHable similarities (\thref{def:LSHability}), and let us denote this by $\mathcal{L}$.  Optionally, consider  $\mathcal{L} \ni S : \mathcal{X} \times \mathcal{X} \rightarrow [0,1]$ where $|\mathcal{X}|$ is a power of 2.

Our first set of open problems address whether a supermodular similarity is necessarily LSHable.  We begin with a more restricted problem (\thref{prob:LSHsymmetricSupermodularHammingSimilarity}) building to a desired general result (\thref{prob:doesSupermodularImplyLSHable}).
\begin{problem}[LSHability of cardinality-based supermodular Hamming similarities?]\thlabel{prob:LSHsymmetricSupermodularHammingSimilarity}
Are all cardinality-based supermodular Hamming similarities LSHable?  If not, what are necessary and sufficient conditions for a cardinality-based supermodular Hamming similarity to be LSHable?
\end{problem}
The background to \thref{prob:LSHsymmetricSupermodularHammingSimilarity} is that cardinality-based supermodular Hamming similarities (CSHSs) are a very special case of supermodular similarities that are parametrized by a single convex function (Section~\ref{sec:SymmetricSupermodularHammingSimilarities}).  Through \thref{thm:LSHpreservingIsSupermodularPreserving} it is straightforward to characterize a rich subset of CSHSs that are LSHable (probability generating functions applied to the Hamming similarity), but it is clear that this subset does not encompass all CSHSs (\thref{thm:SSHMbiggerthanLSHP-H}).  Due to the simple characterization of CSHSs by a 1D convex function, study of this family promises to shed light on the larger problem while making use of the rich set of tools available from convex analysis.  A negative answer to this question will partially answer \thref{prob:doesSupermodularImplyLSHable}, while a positive answer will provide the first results about novel LSHable similarities originating in submodular analysis.

\begin{problem}[When does supermodular imply LSHable?] \thlabel{prob:doesSupermodularImplyLSHable}
For $S$ a metric supermodular similarity (\thref{def:MetricSupermodularSimilarity}), what are necessary and sufficient conditions for $S\in \mathcal{L}$?  Can one find a constructive proof that implies a polynomial time algorithm for a LSH given a metric supermodular similarity (possibly satisfying additional sufficient conditions)?
\end{problem}

Analogous to \thref{prob:doesSupermodularImplyLSHable}, we ask when LSHability implies supermodularity.  This work goes some limited way in answering this question for some specific (families of) similarities (Table~\ref{tab:setSimilarityMeasures}).  However, this does not address the more general characterization of LSHable similarities outside of the limited families analyzed here.  The following set of problems target this by addressing when the LSHability of a similarity implies its supermodularity (arguably a better understood property).  This would potentially add powerful tools to the analysis of LSHable similarities, for which the most useful results currently are necessary (but not sufficient) metric conditions \cite{Charikar2002SET} and that LSHable similarities are $\ell_{1}^{\mathcal{O}(|\mathcal{X}|^2)}$ embeddable \cite[Lemma~5.1]{Chierichetti:2012:LFA:2095116.2095201} .

\begin{problem}[LSHable similarities where $|\mathcal{X}|$ is a power of 2]\thlabel{prob:LSHableSizeUniversePowerOf2}
Consider $S : \mathcal{X} \times \mathcal{X} \rightarrow [0,1]$ such that $|\mathcal{X}| = 2^p$ for some $p \in \mathbb{Z}_+$.  This implies a natural mapping to a set similarity with a base set $V$ where $|V|=p$.  In this setting, what are necessary and sufficient conditions for $S \in \mathcal{L}$ to be a supermodular similarity (\thref{def:SupermodularSimilarity})?
\end{problem}
The background to \thref{prob:LSHableSizeUniversePowerOf2} is our analysis of the intersection similarity \cite[Definition~3.14]{Chierichetti:2012:LFA:2095116.2095201}.  \thref{thm:IntersectionSimilarity} and \thref{thm:IdentityIntersectionSimilarity} together indicate the importance of the specific encoding chosen in converting an LSHable similarity over a more general domain to a set similarity.  One encoding results in a set similarity that is neither submodular nor supermodular, while the other (only valid for $|\mathcal{X}|$ a power of 2) yields a well behaved supermodular similarity.  Indeed, other recent work analyzes ``hidden'' submodularity, and in which cases a transformation of the set encoding exists that maps to a submodular function \cite{DBLP:journals/corr/abs-1712-08721}.  Such analysis is trivially applicable to supermodular functions.

Our analysis in \thref{thm:IntersectionSimilarity} and \thref{thm:IdentityIntersectionSimilarity} exposes interesting dependencies between the encoding and pseudometric and monotonic properties of the resulting set function.  This motivates
\thref{prob:PseudometricLSHableSubmodular}.  We have included a monotonicity condition in our definition of a supermodular similarity (\thref{def:SupermodularSimilarity}, condition~\ref{item:SupermodularSimilarityMonotonicCondition}), but this is used primarily in \thref{lemma:productNonNegMonotonicSupermodular} to show that LSH-preserving functions also preserve supermodularity (\thref{thm:LSHpreservingIsSupermodularPreserving}).  We ask next when monotonicity is actually necessary, and the implications of $1-S$ being only a pseudometric.
\begin{problem}[Pseudometrics and monotonicity]\thlabel{prob:PseudometricLSHableSubmodular}
What are the combined implications of pseudometrics and similarities that are (non)monotonic in the symmetric difference of their arguments?
\end{problem}

\thref{prob:LSHableSizeUniverseMoreGenerally} addresses the setting
for which the domain of $S$ does not have an obvious mapping to a set similarity.  Are there nevertheless strategies for exploiting submodular analysis in constructing an equivalent set similarity?

\begin{problem}[LSHable similarities more generally]\thlabel{prob:LSHableSizeUniverseMoreGenerally}
Consider $S : \mathcal{X} \times \mathcal{X} \rightarrow [0,1]$ such that $|\mathcal{X}| \neq 2^p$ for some $p \in \mathbb{Z}_+$.  This no longer implies a natural mapping to a set similarity with a base set $V$ where $|V|=p$.  In this setting, what are necessary and sufficient conditions for a set encoding of such an $S \in \mathcal{L}$ to be a supermodular set similarity (\thref{def:SupermodularSimilarity})?
\end{problem}

\section{Submodularity}\label{sec:Submodularity}

In this section, we introduce the main mathematical objects from submodular analysis necessary in the sequel.

\begin{definition}[Set function \cite{Schrijver2003combinatorial}]
A set function $\ell$ is a mapping from the power set of a base set $V$ to the reals:
\begin{equation}
\ell : \mathcal{P}(V) \rightarrow \mathbb{R} .
\end{equation}
\end{definition}

\begin{definition}[Submodular set function \cite{fujishige2005submodular}]\thlabel{def:Submodular}
  A set function $f$ is said to be submodular if for all $A \subseteq B \subset V$ and $x\in V\setminus B$,
  \begin{equation}\label{eq:SubmodularDefinitionInequality}
    f(A\cup \{x\}) -f(A) \geq f(B \cup \{x\}) - f(B) .
  \end{equation}
\end{definition}
A set function is said to be supermodular if its negative is submodular, and a function is said to be modular if it is both submodular and supermodular.   
Additional properties of submodular functions can be found in \cite{fujishige2005submodular}.

An alternative characterization of the submodularity of set functions useful in practice is given by the second-order differences \cite[Proposition~1.2]{Bach13fot}.
\begin{theorem}[Submodularity with second-order differences]\thlabel{thm:Submodular2ndOrder}
  A set function $f$ is submodular if and only if for all $A \subset V$ and $s, t\in V\setminus A$,
  \begin{equation}\label{eq:SubmodularSecondOrder}
    f(A\cup \{t\}) -f(A) \geq f(A \cup \{s, t\}) - f(A \cup \{s\}) .
  \end{equation}
\end{theorem}

\section{Submodularity of Set Similarities}\label{sec:SubmodularityofSetSimilarity}

Set similarity measures, such as the Jaccard index, are functions that accept two sets and output (normalized) similarities based on the overlap of elements between these sets, $S : \mathcal{P}(V)^2 \rightarrow \mathbb{R}_+$.  These can also be used to construct dissimilarities e.g.\ by taking one minus the similarity measure for similarity measures normalized between zero and one.  An interesting mathematical property is the submodularity of the measure with respect to the symmetric difference of its arguments.  This has recently been highlighted e.g.\ in the construction of loss surrogates for set prediction \cite{yu:hal-01151823,Berman2018a}.
We will denote the symmetric difference:
\begin{equation}
X \triangle Y := (X\cup Y) \setminus (X\cap Y) = (X\setminus Y) \cup (Y \setminus X).
\end{equation}
In the sequel, we assume that $X$ is fixed, and we will analyze the submodularity of similarity measures including those in Table~\ref{tab:setSimilarityMeasures}.  

\begin{definition}[Similarity]\thlabel{def:similarity}
A function $S : \mathcal{X} \times \mathcal{X} \rightarrow [0,1]$ is called a similarity if
\begin{enumerate}
\item $S(X,X) = 1$;
\item $S(X,Y) = S(Y,X)$.
\end{enumerate}
\end{definition}

For the similarities in Table~\ref{tab:setSimilarityMeasures}, there is a high overlap between supermodularity and LSHability of the measure.  By supermodularity, we mean the following:
\begin{definition}\thlabel{def:SupermodularSimilarity}
A similarity $S$ is said to be supermodular if, holding one argument fixed, the resulting set function of its symmetric difference $f_X: A \mapsto S(X,X\triangle A)$ satisfies the following conditions:
\begin{enumerate}
\item $f_X$ supermodular; \label{symmetricDiffConditionSupermodular}
\item monotonically decreasing, i.e. $f_X(A) \geq f_X(B)$ for all $A \subseteq B$. \label{item:SupermodularSimilarityMonotonicCondition}
\end{enumerate}
\end{definition}
We note that these conditions are equivalent to $1-f_X$ being a polymatroid rank function \cite[Section~2.3]{Bach13fot} with maximum value less than or equal to one.  Furthermore, from the symmetry of $S$, we have that for all $X, Y \subseteq V$,
\begin{align}
f_X(Y) = f_{X\triangle Y}(Y) .
\end{align}

To relate these function values to \thref{def:Submodular}, we consider sets $A, B$ such that $A\subseteq B$ and we denote sets sets $Y := X\triangle A$ and $\tilde{Y} := X\triangle B$, such that $A = X\triangle Y$ and $B = X\triangle \tilde{Y}$.
We have that $|X\cap Y| \geq |X \cap \tilde{Y}|$. We will denote 
$|X \setminus Y| = \alpha$, $|Y \setminus \tilde{Y}| = \beta$, $|X \cap \tilde{Y}| = \zeta$, $|Y \setminus X| = \delta$, and $|\tilde{Y} \setminus Y| = \varepsilon$, 
where in the case that $X\triangle Y \subseteq X \triangle \tilde{Y}$ all Greek variables are unconstrained except that they must be non-negative integers (Figure~\ref{fig:SymmetricDifferenceSupersetFigureGreek}).  We then have the following equalities:
\begin{align}
|X| =& \alpha + \beta + \zeta \label{eq:GreekSizeDefinitionsFirst}\\
|Y| =& \beta + \zeta + \delta\\
|\tilde{Y}| =& \zeta + \delta + \varepsilon \\
|X \cap Y| =& \beta + \zeta\\
|X \cap \tilde{Y}| =& \zeta\\
|X \triangle Y| =& \alpha + \delta\\
|X \triangle \tilde{Y}| =& \alpha + \beta + \delta + \varepsilon\\
|\overline{X \cup Y}| =& |V| - (\alpha + \beta + \zeta + \delta) = \eta + \varepsilon \\
|\overline{X \cup \tilde{Y}}| =& |V| - (\alpha + \beta + \zeta + \delta + \varepsilon) = \eta \label{eq:GreekSizeDefinitionsLast}
\end{align}

\begin{figure}
\centering
\includegraphics[width=0.6\textwidth]{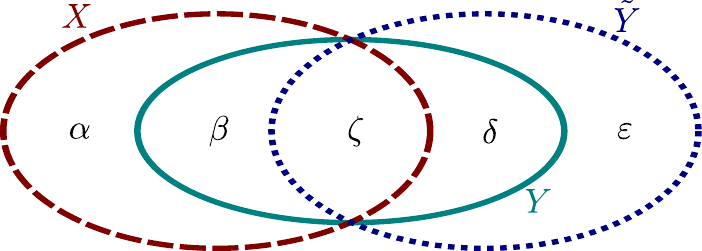}
\caption{Illustration of the sets used in the construction of Equations~\eqref{eq:GreekSizeDefinitionsFirst}-\eqref{eq:GreekSizeDefinitionsLast} (the greek letters represent cardinalities of the corresponding subsets).}\label{fig:SymmetricDifferenceSupersetFigureGreek}
\end{figure}

We begin by considering the set of similarities described in \cite[Table~1]{chierichetti_et_al:LIPIcs:2017:8168}.  These are reproduced here in Table~\ref{tab:setSimilarityMeasures} along with the intersection similarity due to \cite[Definition~3.14]{Chierichetti:2012:LFA:2095116.2095201}.  The Jaccard index has previously been analyzed, specifically the loss version has been shown to be submodular \cite[Proposition~11]{Yu2015b} indicating that the similarity is supermodular.  The Hamming similarity can readily be shown to be modular as it is simply $1-\frac{|X\triangle Y|}{|V|}$ and $|V|$ is fixed.  In the remainder of this section we systematically analyze other similarities showing a one-to-one relationship between a similarity being LSHable and supermodular.

\begin{table}
\caption{Similarity measures from \cite[Table~1]{chierichetti_et_al:LIPIcs:2017:8168} and \cite[Definition~3.14]{Chierichetti:2012:LFA:2095116.2095201}. LSHability results can be found in these references. We note that the cardinality intersection and identity intersection are different set encodings of the same similarity, indicating that supermodularity is a property in part of the encoding of the set similarity while LSHability is rather a property of the underlying metric. In all cases, we have a one to one correspondence between the existence of a supermodular encoding of the similarity and its LSHability.}\label{tab:setSimilarityMeasures}
\centering
\begin{tabular}{p{0.2\textwidth}cp{0.35\textwidth}c}
\textbf{name} & $S(X,Y)$ ($X\neq Y$) & Submodularity w.r.t.\ $X \triangle Y$ & LSHable\\
\hline
Jaccard & $\frac{|X\cap Y|}{|X\cap Y| + |X \triangle Y|}$ & Supermodular \cite[Proposition~11]{Yu2015b} & yes\\
Hamming & $\frac{|X\cap Y| + |\overline{X \cup Y}|}{|X\cap Y| + |\overline{X \cup Y}| + |X \triangle Y|}$ & Modular (Section~\ref{sec:SubmodularityofSetSimilarity}) & yes \\
Anderberg & $\frac{|X\cap Y|}{|X\cap Y| + 2|X \triangle Y|}$ & Supermodular (\thref{thm:SorensenGamma}) & yes \\
Rogers–Tanimoto & $\frac{|X\cap Y| + |\overline{X \cup Y}|}{|X\cap Y| + |\overline{X \cup Y}| + 2|X \triangle Y|}$ & Supermodular (\thref{prop:RogersTanimotoSimilarity}) & yes \\
Simpson & $\frac{|X\cap Y|}{\min(|X|,|Y|)}$ & Neither submodular nor supermodular (\thref{thm:SimpsonSetSimilarity}) & no \\
Braun–Blanquet & $\frac{|X\cap Y|}{\max(|X|,|Y|)}$ & Neither submodular nor supermodular (\thref{thm:BraunBlanquetSimilarity}) & no \\
S{\o}rensen-Dice & $\frac{|X\cap Y|}{|X\cap Y| + \frac{1}{2} |X \triangle Y|}$ & Neither submodular nor supermodular \cite[Proposition~6]{Yu2016a} & no\\
Sokal–Sneath 1 & $\frac{|X\cap Y| + |\overline{X \cup Y}|}{|X\cap Y| + |\overline{X \cup Y}| + \frac{1}{2}|X \triangle Y|}$ & Submodular (\thref{prop:SokalSneath1Similarity}) & no \\
Forbes & $\frac{|V| \cdot |X\cap Y|}{|X|\cdot |Y|}$ &  Neither submodular nor supermodular (\thref{thm:ForbesSimilarity}) & no\\
S{\o}rensen$_\gamma$ & $\frac{|X\cap Y|}{|X\cap Y| + \gamma |X \triangle Y|}$ & Supermodular for $\gamma \geq 1$, neither submodular nor supermodular for $0<\gamma < 1$ (\thref{thm:SorensenGamma}) & iff $\gamma\geq 1$ \\
Sokal–Sneath$_\gamma$ & $\frac{|X\cap Y| + |\overline{X \cup Y}|}{|X\cap Y| + |\overline{X \cup Y}| + \gamma |X \triangle Y|}$ & Supermodular for $\gamma \geq 1$, Submoduar for $0<\gamma<1$ (\thref{prop:SokalSneathGammaSimilarity})  & iff $\gamma \geq 1$\\
Cardinality Intersection & \thref{def:CardinalityIntersectionSimilarity} & Neither submodular nor supermodular (\thref{thm:IntersectionSimilarity}) & yes \\
Identity Intersection & \thref{def:IdentityIntersectionSimilarity} & Supermodular (\thref{thm:IdentityIntersectionSimilarity}) & yes \\
\end{tabular}
\end{table}

\begin{proposition}[S{\o}rensen$_\gamma$ similarity]\thlabel{thm:SorensenGamma}
The S{\o}rensen$_\gamma$ similarity is supermodular for all $\gamma\geq 1$ and is neither submodular nor supermodular for $0<\gamma < 1$. 
\end{proposition}
\begin{proof}
Consider $X$ fixed.  If the symmetric difference is increased by adding another element $y$ to $Y$ that is not in $X$:
\begin{align}
&\frac{|X\cap (Y \cup \{y\})|}{|X\cap (Y \cup \{y\})| + \gamma|X \triangle (Y \cup \{y\})|} 
= \frac{|X\cap Y |}{|X\cap Y| + \gamma|X \triangle Y|+\gamma} 
,
\end{align}
and
\begin{align}\label{eq:SorensenGammaSubmodularityMainStep}
&\frac{|X\cap Y |}{|X\cap Y| + \gamma|A|+\gamma} - \frac{|X\cap Y|}{|X\cap Y| + \gamma|A|}
-
\frac{|X\cap \tilde{Y} |}{|X\cap \tilde{Y}| + \gamma|B|+\gamma} + \frac{|X\cap \tilde{Y}|}{|X\cap \tilde{Y}| + \gamma|B|} = \\
&
\frac{\beta+\zeta}{\beta+\zeta + \gamma(\alpha+\delta)+\gamma} - \frac{\beta+\zeta}{\beta+\zeta + \gamma(\alpha+\delta)}
-
\frac{\zeta}{\zeta + \gamma(\alpha + \beta + \delta + \varepsilon)+\gamma} + \frac{\zeta}{\zeta + \gamma(\alpha + \beta + \delta + \varepsilon)} \nonumber
.
\end{align}
We note that $\beta + \zeta \geq \zeta$ and $\beta + \zeta + \gamma(\alpha + \delta) \leq \zeta + \gamma(\alpha + \beta + \delta + \varepsilon)$.  This means the magnitude of the difference of the first two terms in Equation~\eqref{eq:SorensenGammaSubmodularityMainStep} is greater than the magnitude of the difference of the second two terms.  As the difference of the first two terms is negative, the sum of all four terms must always be negative.  
Equation~\eqref{eq:SorensenGammaSubmodularityMainStep} is the difference between the r.h.s.\ and the l.h.s.\ of Inequality~\eqref{eq:SubmodularDefinitionInequality}.  We therefore determine that the S{\o}rensen$_\gamma$ similarity is non-submodular with respect to the symmetric difference for all $\gamma>0$.

We now consider removing an element $x$ from $Y$ that is in $X$:
\begin{align}
&\frac{|X\cap (Y \setminus \{x\})|}{|X\cap (Y \setminus \{x\})| + \gamma|X \triangle (Y \setminus \{x\})|} 
=\frac{|X\cap Y|-1}{|X\cap Y| + \gamma|X \triangle Y|+\gamma-1} 
.
\end{align}
\begin{align}\label{eq:SorensenGammaSubmodularityFalseNegative}
&\frac{|X\cap Y|-1}{|X\cap Y| + \gamma|A|+\gamma-1} - \frac{|X\cap Y|}{|X\cap Y| + \gamma|A|} - 
\left(
\frac{|X\cap \tilde{Y}|-1}{|X\cap \tilde{Y}| + \gamma|B|+\gamma-1} - \frac{|X\cap \tilde{Y}|}{|X\cap \tilde{Y}| + \gamma|B|}
\right) = \\
&  \frac{\beta+\zeta-1}{\beta+\zeta + \gamma(\alpha+\delta)+\gamma-1} - \frac{\beta+\zeta}{\beta+\zeta + \gamma(\alpha+\delta)} - 
\frac{\zeta-1}{\zeta + \gamma(\alpha+\beta+\delta+\varepsilon)+\gamma-1} + \frac{\zeta}{\zeta + \gamma(\alpha+\beta+\delta+\varepsilon)}
\nonumber
.
\end{align}
For $\gamma \geq 1$, by an analogous argument to that following Equation~\eqref{eq:SorensenGammaSubmodularityMainStep}, we have that the magnitude of the difference of the first two terms is greater than the magnitude of the difference of the second two terms, meaning the sum is negative and the resulting similarity is supermodular with respect to the symmetric difference of its arguments.

For $0<\gamma<1$, assume $\beta=\zeta=1$, and $\delta=\alpha=\varepsilon=0$, then the r.h.s.\ of Equation~\eqref{eq:SorensenGammaSubmodularityFalseNegative} equals
\begin{align}
\frac{2}{1 +\gamma} - 1 >0.
\end{align}
We note that $\beta = \zeta = 1$ satisfies the assumptions that $X\cap Y \neq \emptyset$ and $X\cap \tilde{Y} \neq \emptyset$.
We therefore conclude that the S{\o}rensen$_\gamma$ similarity is neither submodular nor supermodular for $0<\gamma<1$.
\end{proof}

\begin{remark}
\thref{thm:SorensenGamma} subsumes several previous results, including \cite[Proposition~11]{Yu2015b} and \cite[Proposition~6]{Yu2016a}, and implies that the Anderberg similarity is supermodular with respect to the symmetric difference of its arguments (Table~\ref{tab:setSimilarityMeasures}).
\end{remark}

\begin{proposition}[Sokal-Sneath$_\gamma$ similarity]\thlabel{prop:SokalSneathGammaSimilarity}
The Sokal-Sneath$_\gamma$ similarity is supermodular with respect to the symmetric difference of its arguments when $\gamma \geq 1$ and submodular when $0 < \gamma < 1$.
\end{proposition}
\begin{proof}
The Sokal-Sneath$_\gamma$ similarity between $X$ and $Y$ is:
\begin{align}\label{eq:SokalSneathGammaGreekVariables}
\frac{|X\cap Y| + |\overline{X \cup Y}|}{|X\cap Y| + |\overline{X \cup Y}| + \gamma|X \triangle Y|} 
= \frac{|V| - |X \triangle Y|}{|V|  + (\gamma-1)|X\triangle Y|}.
\end{align}

If the symmetric difference is increased by adding another element $y$ to $Y$ that is not in $X$:
\begin{align}\label{eq:SokalSneathGammaFalsePositive}
\frac{|X\cap (Y \cup \{y\})| + |\overline{X \cup (Y \cup \{y\})}|}{|X\cap (Y \cup \{y\})| + |\overline{X \cup (Y \cup \{y\})}| + \gamma|X \triangle (Y \cup \{y\})|} =\\
\frac{|X\cap Y| + |\overline{X \cup Y}|-1}{|X\cap Y| + |\overline{X \cup Y}| -1 + \gamma|X \triangle Y| + \gamma} \nonumber
.
\end{align}

We now consider removing an element $x$ from $Y$ that is in $X$:
\begin{align}
\label{eq:SokalSneathGammaFalseNegative}
\frac{|X\cap (Y \setminus \{x\})| + |\overline{X \cup (Y \setminus \{x\})}|}{|X\cap (Y \setminus \{x\})| + |\overline{X \cup (Y \setminus \{x\})}| + \gamma|X \triangle (Y \setminus \{x\})|} = \\
\frac{|X\cap Y| -1 + |\overline{X \cup Y }|}{|X\cap Y| -1 + |\overline{X \cup Y}| + \gamma|X \triangle Y|+\gamma}
\nonumber
.
\end{align}
We can see that the r.h.s.\ of Equation~\eqref{eq:SokalSneathGammaFalsePositive} is equal to the r.h.s.\ of Equation~\eqref{eq:SokalSneathGammaFalseNegative}.

For $\gamma\geq 1$, the numerator of the Sokal-Sneath$_\gamma$ similarity is monotonically decreasing with the symmetric difference, while the denominator is monotonically increasing with the symmetric difference (see Equation~\eqref{eq:SokalSneathGammaGreekVariables}).  We can therefore apply the same logic as after Equation~\eqref{eq:SorensenGammaSubmodularityMainStep} to conclude that the similarity is supermodular when $\gamma \geq 1$.

For $0<\gamma<1$, 
we will denote $\tilde{\alpha}=\alpha+\delta$ and $\tilde{\beta} = \beta + \varepsilon$:
\begin{align}
  &  \frac{|V| - |X \triangle Y| -1}{|V|  - (1-\gamma)|X\triangle Y| - (1-\gamma)}
  -  \frac{|V| - |X \triangle Y|}{|V|  - (1-\gamma)|X\triangle Y|}\nonumber \\
&  - \left(
\frac{|V| - |X \triangle \tilde{Y}| -1}{|V|  - (1-\gamma)|X\triangle \tilde{Y}| - (1-\gamma)}
  - \frac{|V| - |X \triangle \tilde{Y}|}{|V|  - (1-\gamma)|X\triangle \tilde{Y}|}
  \right) \nonumber \\
=&  \frac{|V| - \tilde{\alpha} -1}{|V|  - (1-\gamma)
    \tilde{\alpha} - (1-\gamma)}
  -  \frac{|V| - \tilde{\alpha}}{|V|  - (1-\gamma)\tilde{\alpha}}
\\ \nonumber &  - \left(
\frac{|V| - (\tilde{\alpha} + \tilde{\beta}) -1}{|V|  - (1-\gamma)(\tilde{\alpha} + \tilde{\beta}) - (1-\gamma)}
- \frac{|V| - (\tilde{\alpha} + \tilde{\beta})}{|V|  -  (1-\gamma)(\tilde{\alpha} + \tilde{\beta})}
\right) \\
=& \frac{
\tilde{\beta}(1-\gamma)\gamma |V|(2 |V| -(1-\gamma)(2\tilde{\alpha}+\tilde{\beta}+1))
}{  (|V|-\tilde{\alpha}(1-\gamma))(|V|-(\tilde{\alpha}+1)(1-\gamma))(|V|-(\tilde{\alpha}+\tilde{\beta})(1-\gamma))(|V|-(\tilde{\alpha}+\tilde{\beta}+1)(1-\gamma))} \label{eq:SokalSneathGammaSubmodularGammaLT1}
\end{align}
For all  $\tilde{\alpha}\geq 0$, $\tilde{\beta}\geq 0$ satisfying $\tilde{\alpha}+\tilde{\beta}< |V|$, and $0<\gamma<1$, each term in the numerator and each term in the denominator of Equation~\eqref{eq:SokalSneathGammaSubmodularGammaLT1} is positive, indicating that the Sokal-Sneath$_\gamma$ similarity is submodular.
\end{proof}

\begin{corollary}[Rogers-Tanimoto similarity]\thlabel{prop:RogersTanimotoSimilarity}
The Rogers-Tanimoto similarity is supermodular with respect to the symmetric difference in its arguments, holding one argument fixed.
\end{corollary}

\begin{corollary}[Sokal-Sneath 1 similarity]\thlabel{prop:SokalSneath1Similarity}
The Sokal-Sneath 1 similarity is submodular with respect to the symmetric difference in its arguments, holding one argument fixed.
\end{corollary}

\begin{proposition}[Simpson similarity]\thlabel{thm:SimpsonSetSimilarity}
The Simpson set similarity is neither submodular nor supermodular with respect to the symmetric difference between its arguments.
\end{proposition}
\begin{proof}
Consider increasing the symmetric difference by adding an element $y$ to $Y$ that is not in $X$:
\begin{align}\label{eq:SimpsonFalsePositiveNoGreek}
\frac{|X\cap (Y\cup \{y\})|}{\min(|X|,|Y\cup \{y\}|)} = \frac{|X\cap Y|}{\min(|X|,|Y|+1)} = \frac{|X\cap Y|}{|X\cap Y|+\min(|X\setminus Y|,|Y\setminus X| +1)}
\end{align}
and
\begin{align}
\frac{|X\cap Y|}{|X\cap Y|+\min(|X\setminus Y|,|Y\setminus X| +1)} - \frac{|X\cap Y|}{|X\cap Y|+\min(|X\setminus Y|,|Y\setminus X|)} - \nonumber\\
\left(
\frac{|X\cap \tilde{Y}|}{|X\cap \tilde{Y}|+\min(|X\setminus \tilde{Y}|,|\tilde{Y}\setminus X| +1)} - \frac{|X\cap \tilde{Y}|}{|X\cap \tilde{Y}|+\min(|X\setminus \tilde{Y}|,|\tilde{Y}\setminus X|)} 
\right) = \nonumber \\
\frac{\beta+\zeta}{\beta+\zeta+\min(\alpha,\delta +1)} - \frac{\beta+\zeta}{\beta+\zeta+\min(\alpha,\delta)} - \nonumber\\
\left(
\frac{\zeta}{\zeta+\min(\alpha+\beta,\delta+\varepsilon +1)} - \frac{\zeta}{\zeta+\min(\alpha+\beta,\delta+\varepsilon)} 
\right) \label{eq:SimpsonFalsePositiveGreek}
\end{align}
Consider $\beta=\alpha=\varepsilon=1$ and $\zeta=\delta=0$, the r.h.s.\ of Equation~\eqref{eq:SimpsonFalsePositiveGreek} becomes
$-\frac{1}{2}$, indicating the Simpson set similarity is non-submodular.

Now consider $\alpha < \delta$ and $\beta > \delta-\alpha + \varepsilon + 1$.  The r.h.s.\ of Equation~\eqref{eq:SimpsonFalsePositiveGreek} simplifies to
\begin{align}
\frac{\zeta}{\zeta+\delta+\varepsilon}
 - 
\frac{\zeta}{\zeta+\delta+\varepsilon +1}
\end{align}
which is positive for all $\zeta > 0$, indicating the Simpson set similarity is neither submodular nor supermodular.
\end{proof}

\begin{proposition}[Braun–Blanquet similarity]\thlabel{thm:BraunBlanquetSimilarity}
The Braun–Blanquet similarity is neither submodular nor supermodular with respect to the symmetric difference between its arguments.
\end{proposition}
\begin{proof}
Following the derivation for the Simpson similarity (Equation~\eqref{eq:SimpsonFalsePositiveNoGreek}), increasing the symmetric difference by adding an element $y \notin X$ to $Y$ changes the value of the Braun-Blanquet similarity to:
\begin{align}
\frac{|X\cap Y|}{|X\cap Y|+\max(|X\setminus Y|,|Y\setminus X| +1)}
\end{align}
and
\begin{align}
\frac{|X\cap Y|}{|X\cap Y|+\max(|X\setminus Y|,|Y\setminus X| +1)} - \frac{|X\cap Y|}{|X\cap Y|+\max(|X\setminus Y|,|Y\setminus X|)} - \nonumber\\
\left(
\frac{|X\cap \tilde{Y}|}{|X\cap \tilde{Y}|+\max(|X\setminus \tilde{Y}|,|\tilde{Y}\setminus X| +1)} - \frac{|X\cap \tilde{Y}|}{|X\cap \tilde{Y}|+\max(|X\setminus \tilde{Y}|,|\tilde{Y}\setminus X|)} 
\right) = \nonumber \\
\frac{\beta+\zeta}{\beta+\zeta+\max(\alpha,\delta +1)} - \frac{\beta+\zeta}{\beta+\zeta+\max(\alpha,\delta)} - \nonumber\\
\left(
\frac{\zeta}{\zeta+\max(\alpha+\beta,\delta+\varepsilon +1)} - \frac{\zeta}{\zeta+\max(\alpha+\beta,\delta+\varepsilon)} 
\right) \label{eq:BraunBlanquetFalsePositiveGreek}
\end{align}
Consider $\beta= 1$ and $\alpha=\delta=\zeta=\varepsilon=0$, the r.h.s.\ of Equation~\eqref{eq:BraunBlanquetFalsePositiveGreek} becomes $-\frac{1}{2}$ indicating the Braun-Blanquet similarity is non-submodular.

Now consider $\alpha>\delta$ and $
\beta < \delta - \alpha + \varepsilon$ and $\varepsilon$ sufficiently large such that $\beta\geq 0$ and $\delta+\varepsilon>\alpha+\beta$.  The r.h.s.\ of Equation~\eqref{eq:BraunBlanquetFalsePositiveGreek} simplifies to
\begin{align}
 \frac{\zeta}{\zeta+\delta+\varepsilon}- 
\frac{\zeta}{\zeta+\delta+\varepsilon +1}
\end{align}
which is positive for all $\zeta>0$, indicating the Braun-Blanquet similarity is neither submodular nor supermodular.
\end{proof}

\begin{proposition}[Forbes similarity]\thlabel{thm:ForbesSimilarity}
The Forbes similarity is neither submodular nor supermodular with respect to the symmetric difference between its arguments.
\end{proposition}
\begin{proof}
Submodularity is closed over multiplication by positive scalars, so it is sufficient to analyze
$\frac{|X\cap Y|}{|X|\cdot |Y|}$.

When adding an element to $y \notin X$ to $Y$:
\begin{align}
\frac{|X\cap (Y\cup \{y\})|}{|X|\cdot |Y\cup \{y\}|} = \frac{|X\cap Y|}{|X|\cdot (|Y|+1)}
\end{align}
and
\begin{align}
\frac{|X\cap Y|}{|X|\cdot (|Y|+1)} -
\frac{|X\cap Y|}{|X|\cdot |Y|} -
\left(
\frac{|X\cap \tilde{Y}|}{|X|\cdot (|\tilde{Y}|+1)} -
\frac{|X\cap \tilde{Y}|}{|X|\cdot |\tilde{Y}|}
\right) = \nonumber \\
\frac{\beta+\zeta}{(\alpha+\beta+\zeta) (\beta+\zeta+\delta+1)} -
\frac{\beta+\zeta}{(\alpha+\beta+\zeta) (\beta+\zeta+\delta)} - \nonumber \\
\left(
\frac{\zeta}{(\alpha+\beta+\zeta) (\zeta+\delta+\varepsilon+1)} -
\frac{\zeta}{(\alpha+\beta+\zeta)(\zeta+\delta+\varepsilon)}
\right)
\end{align}
Consider $\beta = 0$.  We then have
\begin{align}
\frac{\zeta}{(\alpha+\zeta) (\zeta+\delta+1)} -
\frac{\zeta}{(\alpha+\zeta) (\zeta+\delta)} - \nonumber 
\left(
\frac{\zeta}{(\alpha+\zeta) (\zeta+\delta+\varepsilon+1)} -
\frac{\zeta}{(\alpha+\zeta)(\zeta+\delta+\varepsilon)}
\right)
\end{align}
The numerators of the first two terms are the same as the numerators of the second two terms, except for the addition of $\varepsilon$ which can be arbitrarily large.  The sum of the four terms is therefore negative, and we conclude that the Forbes similarity is non-submodular.

Now consider $\alpha=\delta=\varepsilon=0$, $\beta=\zeta=1$, the sum of the four terms is
$\frac{1}{12}$ indicating that the Forbes similarity is neither submodular nor supermodular.
\end{proof}

In the next definition, we define a set similarity equivalent to \cite[Definition~3.14]{Chierichetti:2012:LFA:2095116.2095201}.  \cite[Definition~3.14]{Chierichetti:2012:LFA:2095116.2095201} encodes the universe as a tuple consisting of an element from the power set of a set of size $k$, and an integer $1\leq i\leq n$.  Here, we consider the power set of a set of size $k+n$ and take the cardinality of the last $n$ elements to encode $i$:
\begin{definition}[Cardinality encoding of intersection similarity \cite{Chierichetti:2012:LFA:2095116.2095201}]\thlabel{def:CardinalityIntersectionSimilarity}
For $k,n \in \mathbb{Z}_+$ and $x,h \in \mathbb{R}_+$ such that $0 \leq x \leq x+kh \leq 1$ the cardinality encoding of the intersection similarity $H_{x,h,k,n} : S_{k,n}^2 \rightarrow \mathbb{R}$ is given by
\begin{align}
S_{k,n} := \mathcal{P}(\{1,\dots,k+n\}) ,
\end{align}
and
\begin{align}
H\left(X,Y\right) := \begin{cases}
x + |X \cap Y \cap \{1,\dots,k\}| h & \text{if } (X\cap \{1,\dots,k\}) \neq (Y \cap \{1,\dots,k\}) \lor i \neq j \\
1 & \text{otherwise,}
\end{cases}
\end{align}
where $i=|X \cap \{k+1,\dots,k+n\}|$ and $j=|Y \cap \{k+1,\dots,k+n\}|$.
\end{definition}

\begin{proposition}[Submodularity of the cardinality encoding of intersection similarity]\thlabel{thm:IntersectionSimilarity}
The intersection similarity is neither submodular nor supermodular when taken with respect to the symmetric difference of its arguments.
\end{proposition}
\begin{proof}
When constructing $f_X(A\cup \{b\})$ (cf.\ \thref{def:SupermodularSimilarity}),
we will focus on two cases: 
(i) changing an element of $Y \cap \{k+1,\dots,k+n\}$ makes $|Y \cap \{k+1,\dots,k+n\}| \neq |X \cap \{k+1,\dots,k+n\}|$, and (ii) changing an element of $Y \cap \{k+1,\dots,k+n\}$ makes $|Y \cap \{k+1,\dots,k+n\}| = |X \cap \{k+1,\dots,k+n\}|$.

Assume that $x + |X \cap Y \cap \{1,\dots,k\}| h < 1$ and that $X \cap \{1,\dots,k\} = Y \cap \{1,\dots,k\}$.
Consider $n=4$, $X \cap \{k+1,\dots,k+n\} = \{k+1,k+2\}$, $A = \{k+2\}$, $B = \{k+2,k+3\}$, $x = k+4$:
\begin{align}
f_X(A\cup \{x\}) -f_X(A) =& 1 - (x + |X \cap Y \cap \{1,\dots,k\}| h) \\
> f_X(B \cup \{x\}) - f_X(B) =& (x + |X \cap Y \cap \{1,\dots,k\}| h) - 1 .
\end{align}
We therefore conclude that the cardinality encoding of the intersection similarity is non-submodular.

Now consider $A = \emptyset$, $B = \{k+2\}$, and $x = k+4$:
\begin{align}
f_X(A\cup \{x\}) -f_X(A) =& (x + |X \cap Y \cap \{1,\dots,k\}| h) - 1\\
< f_X(B \cup \{x\}) - f_X(B) =& 1 - (x + |X \cap Y \cap \{1,\dots,k\}| h) .
\end{align}
And we conclude that the cardinality encoding of the intersection similarity is neither submodular nor supermodular.
\end{proof}

The cardinality encoding was chosen in order to map the domain of the intersection similarity from a tuple of a set and an arbitrary positive integer between $1$ and $n$ (as was originally defined in \cite{Chierichetti:2012:LFA:2095116.2095201}) to a set.  If $n$ is constrained to be a power of two, we may consider the following encoding:
\begin{definition}[Identity encoding of intersection similarity \cite{Chierichetti:2012:LFA:2095116.2095201}]\thlabel{def:IdentityIntersectionSimilarity}
For $k\in \mathbb{Z}_+$, $\log_2 n \in \mathbb{Z}_+$  and $x,h \in \mathbb{R}_+$ such that $0 \leq x \leq x+kh \leq 1$ the identity encoding of the intersection similarity $H_{x,h,k,n} : I_{k,n}^2 \rightarrow \mathbb{R}$ is given by
\begin{align}
I_{k,n} := \mathcal{P}(\{1,\dots,k+\log_2 n\}) ,
\end{align}
and
\begin{align}
H\left(X,Y\right) := \begin{cases}
x + |X \cap Y \cap \{1,\dots,k\}| h & \text{if } 
X \neq Y \\
1 & \text{otherwise.}
\end{cases}
\end{align}
\end{definition}
One may verify that \thref{def:CardinalityIntersectionSimilarity} and \thref{def:IdentityIntersectionSimilarity} encode the same similarity (provided $n$ is a power of two) under the equivalence relations
\begin{align}
S_{k,n} \ni A =& \left( A \cap \{1,\dots, k\}, |A \cap \{k+1,\dots, k+n\}| \right) , \\
I_{k,n} \ni A =& \left( A \cap \{1,\dots, k\}, \sum_{i=1}^{\log_2 n} [k+i \in A] \cdot 2^{i-1} \right) ,
\end{align}
and are thus both LSHable by \cite[Lemma~3.15]{Chierichetti:2012:LFA:2095116.2095201}.  We note that \thref{def:IdentityIntersectionSimilarity} yields a pseudometric if $x+kh=1$ (and a metric if $x+kh<1$), which can be verified by considering $X=Y=\{1,\dots,k\}$ and letting $X\setminus \{1,\dots,k\}$ vary arbitrarily, while \thref{def:CardinalityIntersectionSimilarity} always results in a pseudometric.

\begin{proposition}[Submodularity of the identity encoding of intersection similarity]\thlabel{thm:IdentityIntersectionSimilarity}
The identity encoding of the intersection similarity is  supermodular when taken with respect to the symmetric difference of its arguments.
\end{proposition}
\begin{proof}
Consider $X\triangle Y \subset X\triangle \tilde{Y}$ and $d\notin X\triangle \tilde{Y}$.  First assume $X\triangle Y \neq \emptyset$:
\begin{align}
\underbrace{f_X((X\triangle Y)\cup \{d\})}_{=x+h|X \cap Y \cap \{1,\dots,k\} \setminus \{d\}|} -\underbrace{f_X(X\triangle Y)}_{= x+h|X \cap Y \cap \{1,\dots,k\}|} - \underbrace{f_X((X\triangle \tilde{Y}) \cup \{d\})}_{=x+h|X\cap \tilde{Y} \cap \{1,\dots,k\} \setminus \{d\}|} + \underbrace{f_X(X\triangle \tilde{Y})}_{=x+h|X\cap \tilde{Y} \cap \{1,\dots,k\}|} = 0.
\end{align}
Next assume $X\triangle Y = \emptyset$:
\begin{align}
\underbrace{f_X((X\triangle Y)\cup \{d\})}_{=x+h|X \cap Y \cap \{1,\dots,k\} \setminus \{d\}|} -\underbrace{f_X(X\triangle Y)}_{= 1} - \underbrace{f_X((X\triangle \tilde{Y}) \cup \{d\})}_{=x+h|X\cap \tilde{Y} \cap \{1,\dots,k\} \setminus \{d\}|} + \underbrace{f_X(X\triangle \tilde{Y})}_{=x+h|X\cap \tilde{Y} \cap \{1,\dots,k\}|} = \\ \nonumber
x-1+
h\left(|X \cap Y \cap \{1,\dots,k\} \setminus \{d\}|  \underbrace{- |X\cap \tilde{Y} \cap \{1,\dots,k\} \setminus \{d\}| + |X\cap \tilde{Y} \cap \{1,\dots,k\}|}_{=[d\in X]}\right) \\
= x+
h |X \cap Y \cap \{1,\dots,k\} |  -1 \leq 0.
\end{align}
\end{proof}

It therefore cannot be that supermodularity of an LSHable similarity is tied to $1-S$ being a proper metric as \thref{def:CardinalityIntersectionSimilarity} and \thref{def:IdentityIntersectionSimilarity} both yield pseudometrics when $x+kh=1$.  Importantly, \thref{thm:IntersectionSimilarity} makes use of the non-monotonicity of the similarity with respect to the symmetric difference, while \thref{def:IdentityIntersectionSimilarity} yields a similarity that is monotone in the symmetric difference of its arguments (holding one fixed).

\section{Locality Sensitive Hashability and Submodularity\label{sec:LSHS}}

It is of interest to characterize the relationship between submodularity of similarity measures with respect to the symmetric difference of their arguments, and locality sensitive hashability (LSHability) \cite{chierichetti_et_al:LIPIcs:2017:8168,Chierichetti:2012:LFA:2095116.2095201}.  Although LSHability has been studied for more general domains, we will be interested in the case where the domain remains $\mathcal{P}(V)$.

\begin{definition}[LSHability]\thlabel{def:LSHability}
An LSH for a similarity function $S : \mathcal{X} \times \mathcal{X} \rightarrow [0,1]$ is a probability distribution over a set $\mathcal{H}$ of hash functions definied on $\mathcal{P}(V)$  such that $P_{h\in\mathcal{H}}[h(A)= h(B)] = S(A,B)$. A similarity $S$ is LSHable if there is an LSH for $S$.
\end{definition}
In the sequel, we will frequently be concerned with the case that $\mathcal{X}=\mathcal{P}(V)$ for some base set $V$.

As an immediate corollary to \thref{thm:IntersectionSimilarity} and \cite[Lemma~3.15]{Chierichetti:2012:LFA:2095116.2095201}, which states that the intersection similarity is LSHable, we have:
\begin{corollary}\thlabel{thm:EncodingSupermodularOrNot}
Given a similarity $S$, $S$ LSHable does not imply that an encoding of $\mathcal{X}$ as a set leads to a supermodular set similarity following \thref{def:SupermodularSimilarity}. 
\end{corollary}
We note, however, that for $\mathcal{X}$ with size equal to a power of 2, we were able to find an encoding that leads to a supermodular set similarity following \thref{def:SupermodularSimilarity} (\thref{thm:IdentityIntersectionSimilarity}).  Under what conditions such a mapping exists is a central open question about the relationship between LSHability and submodular analysis.

\subsection{Metric Properties of Supermodular Similarities}\label{sec:MetricSupermodularSimilarities}

\begin{theorem}[\cite{Charikar2002SET}]
Let $S$ be a LSHable similarity, $1-S$ is a (pseudo)metric.
\end{theorem}

We analyze in this section whether being a supermodular similarity is sufficient to yield a (pseudo)metric.

\subsubsection{Supermodular Hamming Similarities}
\begin{definition}[Submodular Hamming metric \cite{NIPS2015_5741}]\thlabel{def:SubmodularHammingMetric}
Given a positive, monotone submodular set function $g$ s.t.\ $g(\emptyset)=0$, the corresponding submodular Hamming metric is
$d_g(X,Y) := g(X\triangle Y)$. 
\end{definition}
We see that apart from a normalization such that $d_g : \mathcal{P}(V)^2 \rightarrow [0,1]$, $S(X,Y) = 1-d_g(X,Y)$ is a special case of \thref{def:SupermodularSimilarity} in which there is no conditioning on $X$ in defining $g$.  We will assume in the sequel that $d_g$ is bounded by $1$.  Such normalization is trivial as we assume monotonicity of $g$.

\begin{definition}[Supermodular Hamming similarity]\thlabel{def:SupermodularHammingSimilarity}
A similarity $S$ is called a supermodular Hamming similarity if $S(X,Y) = 1-d_g(X,Y)$ for some submodular Hamming metric $d_g$.
\end{definition}

\begin{theorem}[\cite{NIPS2015_5741}]\thlabel{thm:SupermodularHammingSimilarityMetric}
For a supermodular Hamming similarity $S$, $1-S$ is a (pseudo)metric.
\end{theorem}
\begin{proof}
All properties of a metric are immediate except for the triangle inequality.
Denote $f = 1-g$.
\begin{align}
1-S(X,Z) \leq 1-S(X,Y) + 1-S(Y,Z)  ,
\end{align}
is equivalent to
\begin{align}
\label{eq:HammingMetricInequalitySimilarity}
 f(X\triangle Y) + f(Y \triangle Z) \leq f(X\triangle Z) +1 .
\end{align}
The following generalization of the triangle inequality to the symmetric difference holds
\begin{align}\label{eq:TriangleInequalitySymmetricDifference}
X\triangle Z \subseteq (X\triangle Y) \cup (Y\triangle Z)
\end{align}
which by the monotonicity of $f$ implies
\begin{align}
f(X\triangle Z) \geq f((X\triangle Y) \cup (Y\triangle Z)).
\end{align}
From the supermodularity of $f$
\begin{align}
f(X\triangle Y) + f(Y\triangle Z) \leq \underbrace{f((X\triangle Y)\cup (Y\triangle Z))}_{\leq f(X\triangle Z)}+ \underbrace{f((X\triangle Y)\cap (Y\triangle Z))}_{\leq 1}
\end{align}
and the desired result follows. 
\end{proof}

A supermodular Hamming similarity requires a supermodular function $f$ such that $f(\emptyset) = 1$, $f$ is monotonically decreasing and non-negative.  We may easily construct such an $f$  from an arbitrary supermodular, respectively submodular, function and a non-negative monotonically increasing modular function.  For an arbitrary submodular function, begin by taking its negative to obtain a supermodular function.
\begin{proposition}\thlabel{thm:SupermodularHammingSimilarityConstruction}
For arbitrary supermodular $g$ and  non-negative, modular increasing $m$ such that $g(V) - g(\emptyset) - \sum_{i\in V} \left(g(\{i\}) - g(\emptyset)\right) + m(V) \neq 0$,\footnote{The condition $g(V) - g(\emptyset) - \sum_{i\in V} \left(g(\{i\}) - g(\emptyset)\right) + m(V) \neq 0$ is satisfied whenever $g$ non-modular or $m$ non-zero.}
\begin{align}\label{eq:SupermodularToSupermodularSimilarityCanonical}
f(X) = \frac{g(V\setminus X) - g(\emptyset) - \sum_{i\in V\setminus X} \left(g(\{i\}) - g(\emptyset)\right) + m(V\setminus X)}{g(V) - g(\emptyset) - \sum_{i\in V} \left(g(\{i\}) - g(\emptyset)\right) + m(V)}
\end{align}
is supermodular, monotonically decreasing, non-negative, and has the property that $f(\emptyset)=1$. Furthermore, all such functions having these properties can be obtained by Equation~\eqref{eq:SupermodularToSupermodularSimilarityCanonical}.
\end{proposition}
\begin{proof}
\emph{Supermodularity:} $g$ is supermodular, and Equation~\ref{eq:SupermodularToSupermodularSimilarityCanonical} only multiplies by positive scalars, adds modular functions, and applies a reflection (i.e.\ replacing $X$ with $V\setminus X$).  Supermodularity is closed under each of these operations.

\emph{Non-negativity and monotonicity:} We first subtract $g(\emptyset)$ to canonically normalize the set function.  $g(X) - g(\emptyset) - \sum_{i\in X} \left(g(\{i\}) - g(\emptyset)\right)$ is supermodular and equal to zero for $X$ equal to any singleton set $\{i\}$,
which by \cite[Lemma~1]{Blaschko2016a} implies the expression is monotonically increasing.  Summation with a monotonically increasing $m$ maintains this property.  Division by $g(V) - g(\emptyset) - \sum_{i\in V} \left(g(\{i\}) - g(\emptyset)\right) + m(V)$ ensures all values are bounded between zero and one.  This remains unchanged when applying a reflection, and the reflection of a monotonically increasing  set function is a monotonically decreasing set function. 

Moreover
\begin{equation}f(\emptyset) = \frac{g(V\setminus \emptyset) - g(\emptyset) - \sum_{i\in V\setminus \emptyset} \left(g(\{i\}) - g(\emptyset)\right) + m(V\setminus \emptyset)}{g(V) - g(\emptyset) - \sum_{i\in V} \left(g(\{i\}) - g(\emptyset)\right) + m(V)} = 1;
\end{equation}
all supermodular, monotonically decreasing, non-negative functions with $f(\emptyset)=1$ can be obtained by Equation~\eqref{eq:SupermodularToSupermodularSimilarityCanonical} for some $g$ and $m$: Let $\hat{f}$ be such a function.  Define $\hat{m}(\emptyset) = \hat{f}(V)$ and $\hat{m}(\{i\}) = \hat{f}(V\setminus \{i\})$, $\forall i$. This uniquely determines the modular function $\hat{m}$ \cite[Equation~(44.2)]{Schrijver2003combinatorial}.  Next define $\hat{g}(X) = \hat{f}(V\setminus X) - \hat{m}(X) \implies \hat{f}(X) = \hat{g}(V\setminus X) + \hat{m}(V\setminus X)$.
\begin{align}
\frac{\hat{g}(V\setminus X) - \overbrace{\hat{g}(\emptyset)}^{=0} - \sum_{i\in V\setminus X} (\overbrace{\hat{g}(\{i\})}^{=0} - \overbrace{\hat{g}(\emptyset)}^{=0}) + \hat{m}(V\setminus X)}{\hat{g}(V) - \hat{g}(\emptyset) - \sum_{i\in V} \left(\hat{g}(\{i\}) - \hat{g}(\emptyset)\right) + \hat{m}(V)} = \frac{\hat{g}(V\setminus X) + \hat{m}(V\setminus X)}{\underbrace{\hat{g}(V) + \hat{m}(V)}_{=\hat{f}(\emptyset)=1}}=\hat{f}(X).
\end{align}
\end{proof}

\begin{remark}
We note that the family of supermodular Hamming similarities is substantially smaller than the family of supermodular similarities following \thref{def:SupermodularSimilarity}.  We may observe for example that the Jaccard index complies with \thref{def:SupermodularSimilarity} but not with \thref{def:SupermodularHammingSimilarity}.
\end{remark}
We therefore consider next similarities that comply with the more general \thref{def:SupermodularSimilarity}.

\subsubsection{Metric Properties of General Supermodular Similarities}

\begin{proposition}\thlabel{thm:SupermodularityNotImpliesMetric}
That a similarity $S$ satisfies \thref{def:SupermodularSimilarity} does not imply $1-S$ is a (pseudo)metric.
\end{proposition}
\begin{proof}
We demonstrate a counterexample for $|V|=2$. Overloading notation so that $S \in \mathbb{R}^{4\times 4}$, and indexing the rows and columns such that $S(X,Y) = S_{i, j}$ where $i = 1+[1\in X] + 2\cdot [2\in X]$ and $j=1+[1\in Y] + 2\cdot [2\in Y]$,
\begin{align}
S = \begin{pmatrix}
  1 & \gamma & \gamma  & \gamma\\
   \gamma &  1 &  0 &  2 \gamma\\
   \gamma & 0 &  1  & 1 - \gamma \\
   \gamma & 2 \gamma & 1 - \gamma &  1
\end{pmatrix},
\end{align}
and for a given $\gamma \in [0, 1/3]$ 
we show that $S$ is symmetric, monotonic, and that each column is supermodular following \thref{def:SupermodularSimilarity}:
\begin{align}
f_{X}(\emptyset) =& 1,\ \forall X, \\
f_{\emptyset}(\{1\}) =& f_{\emptyset}(\{2\}) = f_{\emptyset}(\{1,2\})= \gamma,\\
f_{\{1\}}(\{1\}) =& \gamma,\ f_{\{1\}}(\{2\}) =  2\gamma,\ f_{\{1\}}(\{1,2\}) =   0, \\
   f_{\{2\}}(\{1\})=& 1 - \gamma,\ f_{\{2\}}(\{2\})=  \gamma,\ f_{\{2\}}(\{1,2\})=  0\\
  f_{\{1,2\}}(\{1\}) =& 1 - \gamma,\ f_{\{1,2\}}(\{2\}) = 2 \gamma,\ f_{\{1,2\}}(\{1,2\}) = \gamma.
\end{align}

However, as soon as $\gamma$ is positive, $1-S$ does not satisfy the triangle inequality: indeed, we then have 
$(1-S(\{1\},\{2\})) - (1-S(\{1\},\{1,2\})) - (1-S(\{1,2\},\{2\})) = 
1 - (1 - 2\gamma) - \gamma = \gamma > 0$.
\end{proof}

As \thref{def:SupermodularSimilarity} does not imply that the similarity yields a metric, we propose the following definition that explicitly enforces the triangle inequality:
\begin{definition}[Metric supermodular similarity]\thlabel{def:MetricSupermodularSimilarity}
A metric supermodular similarity is a supermodular similarity (\thref{def:SupermodularSimilarity}) that additionally satisfies that for all $X,Y,Z\subseteq V$:
\begin{align}
 f_X(X\triangle Y) + f_Y(Y\triangle Z) \leq f_X(X\triangle Z) +1 .
\end{align}
\end{definition}
As this metric property is necessary, but not sufficient for LSHability, it remains to be demonstrated under what conditions a similarity being a metric supermodular similarity is sufficient to guarantee LSHability.

\subsection{LSH-Preserving Functions are Supermodularity-Preserving Functions}\label{sec:LSHpreservingIsSupermodularPreserving}

It is well known that the class of LSH-preserving functions is the set of probability generating functions \cite[Section~3]{Chierichetti:2012:LFA:2095116.2095201}.

\begin{definition}[LSH-preserving function]
A function $f : [0,1) \rightarrow [0,1]$ is LSH-preserving if $f \circ S$ is LSHable whenever $S$ is LSHable.
\end{definition}

\begin{definition}[Probability generating function]
A function $f(x)$ is a probability generating function (PGF) if there is a probabilty distribution $\{p_i\}_{0\leq i < \infty}$ such that $f(x) = \sum_{i=0}^{\infty} p_i x^i$ for $x\in [0,1]$.
\end{definition}

\begin{theorem}[Theorem 3.1 \cite{Chierichetti:2012:LFA:2095116.2095201}]\thlabel{thm:lshPreserving}
A function $f : [0,1) \rightarrow [0,1]$ is LSH-preserving iff there are a PGF $p$ and a scalar $\alpha \in [0,1]$ such that $f(x) = \alpha p(x)$.
\end{theorem}

We now show that LSH-preserving functions are supermodularity-preserving functions.

\begin{lemma}\thlabel{lemma:productNonNegMonotonicSupermodular}
Let $f$ and $g$ be two non-negative supermodular functions, both non-increasing or both non-decreasing. The product function $fg$ is supermodular. 
\end{lemma}
\begin{proof}
Let $A \subset V$. For any set function $f$, and an element $t \in A$, we denote
\begin{equation}D_t f = f(A \cup \{t\}) - f(A)\end{equation} the first-order difference of the set function, and 
\begin{equation}P_t f = f(A \cup \{t\}).\end{equation}

Note for the remaining of the proof that operators $(D_u)_{u \in V}$, $(P_v)_{v \in V}$ are distributive over addition and that their application is associative and commutative.

Some easy arithmetic shows that the second-order condition of submodularity~\eqref{eq:SubmodularSecondOrder} is equivalent to
\begin{equation}
f \text{ submodular } \iff D_t D_s f \leq 0 \quad \forall{A \subset V} \quad \forall{s, t \in V \setminus A}
\end{equation}
and similarly, supermodularity reduces to the condition $D_t D_s f \geq 0$. 

One can easily prove the product rule
\begin{equation}
D_s(fg) = (D_s f)(P_s g) + (f)(D_s g).
\end{equation}
Applying this formula a second time (to the products in the RHS), we get that
\begin{equation}\label{eq:secondOrderProduct}
D_t D_s(fg) = D_t D_s f P_t P_s g + D_s f D_t P_s g + D_t f P_t D_s g + f D_t D_s g.
\end{equation}
Supposing $f$ and $g$ are non-negative, supermodular and both nondecreasing or both nonincreasing. Then we can show that each of the terms in the RHS of~\eqref{eq:secondOrderProduct} are positive, which shows the second-order condition of supermodularity of $fg$. Indeed,
\begin{itemize}
\item $(D_t D_s f) (P_t P_s g) \geq 0$ by supermodularity of $f$ and non-negativity of $g$
\item $(D_s f) (D_t P_s g)$ is a product of first-order differences, which have same sign by assumption of shared monoticity; the same holds for $(D_t f) (P_t D_s g)$
\item $f (D_t D_s g) \geq 0$ by non-negativity of $f$ and supermodularity of $g$.
\end{itemize}
\end{proof}

\begin{proposition}[LSH-preserving functions are supermodularity-preserving functions]\thlabel{thm:LSHpreservingIsSupermodularPreserving}
Given an LSH-preserving function $f : [0,1) \rightarrow [0,1]$ and a non-negative monotonically decreasing supermodular function $g$ such that $g(\emptyset) = 1$, $f\circ g$ is a non-negative monotonically decreasing supermodular function with $f \circ g(A) \in [0,1]$ for all $A\subseteq V$.
\end{proposition}
\begin{proof}
The following properties are closed under convex combination:
supermodularity, monotonicity, non-negativity.  By definition we require that $f(\emptyset)=1$, which is why the definition of a LSH-preserving function is over the domain $[0,1)$.

It therefore only remains to show that if $g$ satisfies the above conditions, that $g^i$ for an arbitrary non-negative integer also satisfies the desired properties.
Monotonicity and non-negativity is closed under non-negative exponentiation.
That supermodularity is preserved under non-negative integer exponentiation is a straightforward consequence of \thref{lemma:productNonNegMonotonicSupermodular}.
\end{proof}

\thref{thm:LSHpreservingIsSupermodularPreserving} implies that functions that are both supermodular with respect to the symmetric difference and LSHable will remain so after LSH-preserving transformations.  This indicates that there is at least a special class of supermodular LSHable functions that is closed under such transformations.  

\subsection{Cardinality-Based Supermodular Hamming Similarities}\label{sec:SymmetricSupermodularHammingSimilarities}

In Sections~\ref{sec:SubmodularityofSetSimilarity} and~\ref{sec:LSHS}, we have put into light various connections between metric supermodular similarities and LSHable similarities. 
In particular, the intersection between these two families of set functions contains many common supermodular similarities, as summarized in Table~\ref{tab:setSimilarityMeasures}. 
Moreover, LSH-preserving transformations also preserve the supermodularity and metricity of supermodular metric similarities (\thref{thm:LSHpreservingIsSupermodularPreserving}).
We are however far from a complete characterization of this intersection. 
In this section we consider a more restricted setting where the analysis of submodularity reduces to convex analysis. 

\begin{definition}[Cardinality-based set functions]\thlabel{def:CSF}
A set function $F$ over a base set $V$ is said to be cardinality-based if there exists $g: \mathbb{R}_+ \rightarrow \mathbb{R}$ such that 
$
F(A) = g(|A|)
$ for all $A \in \mathcal{P}(V)$.
\end{definition}

\begin{proposition}[Supermodular cardinality-based set functions (proof similar to {\cite[Prop.~6.1]{Bach13fot}})]\thlabel{thm:supermodularConvex}
A cardinality-based set function $A \mapsto g(|A|)$ is supermodular iff $g$ is convex.
\end{proposition}

Following this observation, we consider the following restricted subset of supermodular metric similarities: 
\begin{definition}[Cardinality-based supermodular Hamming similarities]\thlabel{thm:CSHM}A similarity $S$ is in the set $\mathit{CSHS}$ of cardinality-based supermodular Hamming similarities iff it can be written as $S(X, Y) = h(|X\triangle Y|)$ where $h$ is positive, non-increasing, such that $h(0) = 1$, and convex. 
\end{definition}

Following \thref{thm:supermodularConvex} it is clear that the set $\mathit{CSHS}$ is a subset of the set of supermodular Hamming similarities defined by \thref{def:SupermodularHammingSimilarity}.

We see that the Hamming similarity
\begin{equation}\label{def:hamming}
H(X, Y) = 1-\frac{|X\triangle Y|}{|V|} = h(|X\triangle Y|)\quad \text{where } h(x) = 1 - \frac{x}{V}
\end{equation}
is in $\mathit{CSHS}$.
Therefore, the set $\mathit{LSHP\circ H}$ of LSH-preserving transformations applied to the Hamming similarity is an LSHable subset of $\mathit{CSHS}$. 
The following proposition shows that this is a strict inclusion:
\begin{proposition}\thlabel{thm:SSHMbiggerthanLSHP-H}
The set CSHS of cardinality-based supermodular Hamming similarities is strictly larger than the set of LSH-preserving functions composed with Hamming $\mathit{LSHP\circ H}$.
\end{proposition}
\begin{proof}
Consider
\begin{equation}
f(x) = \frac{3}{2}x^2 - \frac{1}{2}x^3.
\end{equation}
Having negative coefficients in its series expansion, $f$ is not a PGF, and therefore not LSH-preserving following  \thref{thm:lshPreserving}. 
However $f \circ H$, where $H$ is the Hamming similarity, is a cardinality-based supermodular Hamming similarity: with $h$ defined as in~\eqref{def:hamming}, $f\circ h$ is convex by composition, decreasing, positive, and $f(h(0)) = 1$.
\end{proof}

\subsection*{Acknowledgments}

This work is funded by Internal Funds KU Leuven, an Amazon Research Award, and the Research Foundation~--~Flanders (FWO) through project number G0A2716N.  The genesis of this work was the 2017 Data-driven Algorithmics meeting in Bertinoro.  We thank Andreas Krause and Yaron Singer for its organization, and all participants for stimulating discussions on LSHability and submodular analysis.

\bibliographystyle{abbrv}
\bibliography{biblio}

\end{document}